\documentclass[11pt]{article} 
\usepackage[dvipsnames,svgnames,x11names]{xcolor}
\usepackage[utf8]{inputenc}
\usepackage[T1]{fontenc}
\usepackage{authblk}
\usepackage{fullpage}
\usepackage{tabularx}
\usepackage{amsthm}
\usepackage{amsmath}
\usepackage{amssymb}
\usepackage{amsfonts}
\usepackage{mathtools}
\usepackage{enumitem}
\usepackage[pagebackref]{hyperref}
\usepackage[sort,numbers]{natbib}
\usepackage[capitalize]{cleveref}
\usepackage{todonotes}
\usepackage{comment}

\usepackage{tikz}

\usetikzlibrary{arrows,decorations.pathmorphing,decorations.pathreplacing,backgrounds,positioning,fit,matrix}
\usetikzlibrary{shapes,calc,patterns,arrows.meta}
\tikzset{
	vert/.style={circle,inner sep=1.5,fill=white,draw=black,minimum size=.3cm},
	vert2/.style={inner sep=1.5,fill=white,draw=black,minimum size=.3cm},
    dummy/.style={circle,fill=black,draw=black,inner sep=2.5},
	edge/.style={color=black, line width=1pt},
	diredge/.style={->,>={Stealth[width=8pt,length=8pt]},color=black, line width=1pt},
	timelabel/.style={fill=white,font=\footnotesize, text centered},
	wave/.style={decorate,decoration={coil,aspect=0}},
	dirwave/.style={->, >={Stealth[width=8pt,length=8pt]},decorate,decoration={coil,aspect=0}},
	diredge2/.style={->,>={Stealth[width=8pt,length=8pt]}}
}

\usepackage{contour}

\usepackage{multirow}
\usepackage{booktabs}
\usepackage{makecell}
\usepackage{caption}
\usepackage{xspace}

\newtheorem{theorem}{Theorem}
\newtheorem{lemma}[theorem]{Lemma}
\newtheorem{observation}[theorem]{Observation}

\newtheorem{proposition}[theorem]{Proposition}

\theoremstyle{definition}

\usepackage{soul}

\crefname{figure}{Figure}{Figures}

\newcommand{\commentout}[1]{}

\newcommand{\defbox}[2]{
	\begin{center}\fbox{
	\begin{minipage}{0.95\textwidth}
		\vspace{1pt}
 
		\noindent
		\textbf{#1}
  
		\vspace{6pt}
  \centering
  \fcolorbox{lightgray!40!white}{lightgray!40!white}{
	\begin{minipage}{0.95\textwidth}
 \begin{minipage}{0.99\textwidth}
		#2
  \end{minipage}
\end{minipage}}
  \vspace{3pt}
	\end{minipage}}
	\end{center}
}

\newcommand{\problemdef}[3]{
	\begin{center}\fbox{
	\begin{minipage}{0.95\textwidth}
		\vspace{1pt}
 
		\noindent
		\textbf{#1}
  
		\vspace{5pt}
  
  \begin{minipage}{0.99\textwidth}
		\setlength{\tabcolsep}{3pt}
		\begin{tabularx}{\textwidth}{@{}lX@{}}
			\textrm{Input:}     & #2 \\
			\textrm{Question:}  & #3
		\end{tabularx}
  \end{minipage}
	\end{minipage}}
	\end{center}
}

\newcommand{\problembox}[2]{
	\begin{center}\fbox{
	\begin{minipage}{0.95\textwidth}
		\vspace{1pt}
 
		\noindent
		\textbf{#1}
  
		\vspace{5pt}
  
  \begin{minipage}{0.99\textwidth}
  \vspace{1pt}
		 #2
  \end{minipage}
	\end{minipage}}
	\end{center}
}

\newcommand{\OO}{\mathcal{O}}
\newcommand{\N}{\mathbb{N}}
\newcommand{\FPTNP}{\textsf{FPT}$^\textsf{NP}$\xspace}

\newcommand{\FPTNPfew}{\textsf{FPT}$^\textsf{NP}$\textsf{[few]}\xspace}
\newcommand{\PNP}{\textsf{P}$^\textsf{NP}$\xspace}

\newcommand{\OKernel}{\PNP-kernel\xspace}
\newcommand{\OKernels}{\PNP-kernels\xspace}

\newcommand{\paraNP}{\textsf{para-NP}\xspace}
\newcommand{\paracoNP}{\textsf{para-coNP}\xspace}
\newcommand{\NP}{\textsf{NP}\xspace}
\newcommand{\SigmaPtwo}{$\mathsf{\Sigma_2^\textsf{P}}$\xspace}
\newcommand{\PiPtwo}{$\mathsf{\Pi_2^\text{P}}$\xspace}
\newcommand{\coNP}{\textsf{coNP}\xspace}
\newcommand{\FPT}{\textsf{FPT}\xspace}

\newcommand{\wx}[1]{\textsf{W[#1]}\xspace}
\newcommand{\wone}{\wx{1}}
\newcommand{\cowone}{\textsf{co}\wx{1}}

\newcommand{\polyadvice}{{\textsf{poly}}\xspace}
\newcommand{\PSPACE}{\textsf{PSPACE}\xspace}

\newcommand{\NoKernelAssume}{\coNP $\subseteq$ \NP/\polyadvice}
\newcommand{\NNoKernelAssume}{\coNP $\not\subseteq$ \NP/\polyadvice}

\newcommand{\NoOKernelAssume}{\PiPtwo $\subseteq$ \SigmaPtwo/\polyadvice}

\newcommand{\yes}{yes\xspace}
\newcommand{\no}{no\xspace}
\newcommand{\truevalue}{true\xspace}
\newcommand{\falsevalue}{false\xspace}

\begin{document}

\title{On Kernelization with Access to NP-Oracles}


\author{Hendrik~Molter\thanks{Supported by the European Union's Horizon Europe research and innovation program under grant agreement~949707.}~}

\author{Meirav~Zehavi}

\affil{\small Department of Computer Science, Ben-Gurion~University~of~the~Negev, 
Beer-Sheva, 
Israel\\ \texttt{molterh@post.bgu.ac.il, meiravze@bgu.ac.il}}

\date{}

\maketitle
\begin{abstract}
\emph{Kernelization} is the standard framework to analyze preprocessing routines mathematically. Here, in terms of efficiency, we demand the preprocessing routine to run in time polynomial in the input size. However, today, various \NP-complete problems are already solved very fast in practice; in particular, SAT-solvers and ILP-solvers have become extremely powerful and used frequently. Still, this fails to 
capture the wide variety of computational problems that lie at higher levels of the polynomial hierarchy. Thus, for such problems, it is natural to relax the definition of kernelization to permit the preprocessing routine to make polynomially many calls to a SAT-solver, rather than run, entirely, in polynomial time. 

Our conceptual contribution is the introduction of a new notion of a kernel that harnesses the power of SAT-solvers for preprocessing purposes, and which we term a {\em \OKernel}. Technically, we investigate various facets of this notion, by proving both positive and negative results, including a lower-bounds framework to reason about the negative results. Here, we consider both satisfiability and graph problems. Additionally, we present a meta-theorem for so-called ``discovery problems''. This work falls into a long line of research on extensions of the concept of kernelization, including \emph{lossy kernels} [Lokshtanov et al.,~STOC~'17], \emph{dynamic kernels} [Alman et al.,~ACM TALG~'20], \emph{counting kernels} [Lokshtanov et al.,~ICTS~'24], and \emph{streaming kernels} [Fafianie and Kratsch,~MFCS~'14].

\smallskip


\end{abstract}

\section{Introduction}
Effective and efficient preprocessing is the first step of almost any application in modern life. Nowadays, \emph{kernelization} is the standard framework to analyze preprocessing routines mathematically. Accordingly, kernelization was termed ``the lost continent of polynomial time''~\cite{Fellows06}.  Being a subarea of parameterized complexity, kernelization deals with parameterized problems, where every problem instance is associated with a {\em parameter}. The parameter can be any numerical measure of the input or output, such as the sought solution size, the number of variables in a satisfiability/polynomial-based problem, the maximum degree of the graph or its treewidth in a graph problem, the size of the alphabet or the number of strings in a problem involving strings, etc. The most central notion in the context of parameterized problems is that of {\em fixed-parameter tractability (\FPT)}. Formally, a parameterized problem is said to be in \FPT if there exists an algorithm (termed a {\em fixed-parameter algorithm}) that, given an instance $(I,k)$ of the problem, solves it in time $f(k)\cdot |I|^{\OO(1)}$ for some computable function $f$ of $k$.\footnote{We give more formal definitions in \cref{sec:prelims}.} Intuitively, containment in \FPT means that the combinatorial explosion in the running time, which is presumably unavoidable for \NP-hard problems, can be confined to the parameter.

The second most central notion in the context of parameterized problems is that of a {\em kernel}. Formally, we say that a parameterized problem admits a kernel if there exists a polynomial-time algorithm (termed a {\em kernelization algorithm}) that, given an instance $(I,k)$ of the problem, outputs another equivalent instance $(I',k')$ of the problem, such that $|I'|+k'\leq f(k)$ for some computable function~$f$ of $k$. Intuitively, the admittance of a kernel means that every instance of the problem can be preprocessed ``very fast'' (in polynomial time) so that its size is compressed to depend only on $k$.\footnote{Clearly, if we already have that $|I'|+k'\leq f(k)$, then, possibly, no size reduction is performed. However, for an \NP-hard (decidable) problem, we simply cannot expect to always be able to decrease the instance size in polynomial time even by a single bit! Indeed, then we would have been able to just solve the entire instance in polynomial time by repeated applications of the reduction. This anomaly is, in fact, considered to be the main reason why preprocessing routines have not been analyzed in a mathematically rigorous fashion before the advent of parameterized complexity~\cite{Fom+19}.}
Having a particular $f$ in mind, a kernel is termed an $f(k)$-kernel. Obviously, the smaller $f$ is, the better the data compression is. In particular, we are interested in cases where $f$ is a polynomial in $k$---then, a kernel is termed a {\em polynomial kernel}. It is known that a decidable parameterized problem is in \FPT if and only if it admits a kernel~\cite{cai1997advice}. However, unless the polynomial hierarchy collapses, it is long known that there exist parameterized problems that are in \FPT but which do not admit a {\em polynomial} kernel~\cite{BJK14,Bod+09,FS11}. On the positive side, over the past three decades, numerous powerful techniques to develop polynomial kernels as well as meta-theorems to assert their existence were discovered. We refer to the textbook by \citet{Fom+19}, dedicated to kernelization, for more information on this topic. For additional surveys on kernelization, we refer to~\cite{Lokshtanov11,Kratsch14,LokshtanovMS12,GuoN07}, and for additional books on parameterized complexity in general, we refer to~\cite{van2019cognition,Nie06,DF13,FG06,Cyg+15}.

\paragraph{What Is the Meaning of ``Efficiency''?}
The classic definition (given above) of kernelization addresses the requirement of efficiency from the preprocessing routine by restricting it to run in polynomial time. However, today, various \NP-complete problems are already solved very fast in practice. In particular, SAT-solvers~\cite{SakallahM11,biere+21,gomes2008satisfiability,MalikZ09,gong2017survey,alouneh2019comprehensive} and ILP-solvers~\cite{lenstra1983integer,schrijver1998theory,chen2011applied,taha2014integer,wolsey2014integer,junger200950} have become extremely powerful and used frequently.\footnote{For more information on the rapid developments on this subject, we refer the reader to the proceedings of the {\em International Conference on Theory and Applications of Satisfiability Testing} and the {\em International Conference on Integer Programming and Combinatorial Optimization}. }  Thus, the user/researcher only needs to encode the \NP-complete problem at hand to SAT/ILP, and then apply a SAT/ILP-solver on the result. 
Still, this fails to 
capture the wide variety of computational problems that lie at higher levels of the polynomial hierarchy. 
Accordingly, in the fundamental work of~\citet{de2019parameterized} and~\citet{de2017parameterized,de2014fixed,de2019compendium,HaanS15}, it was already suggested to harness the power of SAT/ILP-solvers in the context of fixed-parameter tractability. In particular, they defined the class \FPTNP to consist of all parameterized problems that can be solved by a fixed-parameter algorithm that has oracle access to an \NP-complete problem. Since then, membership in \FPTNP\ was asserted or refuted for problems in various research domains, such as planning~\cite{HaanKP15}, belief revision~\cite{PfandlerRWW15}, reasoning~\cite{HaanS14,SaikkoWJ16,HaanS16}, and agenda safety and judgment aggregation~\cite{EndrissHS15,Haan16}.  

\paragraph{Kernelization in the Presence of SAT/ILP-Solvers.}
We observe that for parameterized problems classified above \NP in the polynomial hierarchy, it is very natural to harness the power of SAT/ILP-solvers in the context of kernelization as well. Specifically, we suggest to
{\bf relax the definition of kernelization to permit the preprocessing routine to make polynomially many calls to a SAT/ILP-solver rather than run entirely in polynomial time}. 

Accordingly, we introduce a new notion of a kernel, termed a {\em \OKernel}. Formally, we say that a parameterized problem admits a \OKernel if there exists a polynomial-time algorithm having oracle access to an \NP-complete problem (such as \textsc{Satisfiability} or \textsc{Integer Linear Programming}) that, given an instance $(I,k)$ of the problem, outputs another instance $(I',k')$ of the problem, such that $|I'|+k'\leq f(k)$ for some computable function $f$ of $k$. Analogously to the case of standard kernelization, we show that a (decidable) parameterized problem is in \FPTNP\  if and only if it admits a \OKernel. Thus, again analogously to the case of standard kernelization, the main question for a given parameterized problem in \FPTNP\ is whether it admits a {\em polynomial} \OKernel. Obviously, the answer is trivially \yes for parameterized problems in \NP or which admit a (standard) polynomial kernel. However, as stated earlier, there exists a wide range of problems for which neither is true. We further discuss the notion of a \OKernel\ in \cref{sec:orableKernels}. 

We remark that \citet{de2019parameterized} discusses ``non-deterministic kernelization'' in Chapter 16 as a future research direction. Here, the kernelization algorithm runs in non-deterministic polynomial time. 
This is somewhat similar to a kernelization algorithm that may use an \NP-oracle at most once (this is not a definition, there is a subtle difference, but it conveys the intuition). We give a brief discussion in \cref{sec:orableKernels} but refer for details to \citet{de2019parameterized}.
 Additionally, \citet{de2019parameterized} provides discussions about some restricted variants, such as non-deterministic polynomial kernelization.
To the best of our knowledge, the work by \citet{de2019parameterized} is the only prior work that considers kernelization algorithms with oracle access to an \NP-complete problem.

\paragraph{Our Contributions.}
Our main goal lies in providing new tools for obtaining parameterized algorithms that exploit the power of SAT/ILP-solvers.
As discussed above, we make the following two conceptual contributions:
\begin{itemize}
    \item We lift the concepts of kernels to the oracle setting by formally introducing (polynomial) \OKernels\ (\cref{sec:orableKernels}).
    \item 
    We expand the framework based on OR-cross-compositions to refute the existence of polynomial kernels~\cite{Bod+09,FS11,BJK14,Fom+19,DellM14} to be applicable for refuting the existence of polynomial \OKernels\ as well (\cref{sec:lower}).
\end{itemize}
 Then, we turn to prove or refute the existence of polynomial \OKernels\ for three different (individual or classes of) parameterized problems. Towards this, we first assert that the investigation of polynomial \OKernels\ for these specific problems is of interest. This is the case if the following situation is given:
\begin{itemize}
    \item The (unparameterized) problem is unlikely to be contained in \PNP, since otherwise it can be solved directly by the kernelization algorithm and hence trivially admits a polynomial \OKernel. This can be done by showing hardness for a complexity class beyond \NP, such as e.g.\ \SigmaPtwo, \PiPtwo, or \PSPACE.\footnote{We give formal definitions in \cref{sec:prelims}.}
    \item The problem is unlikely to admit a (standard) polynomial kernel, since otherwise it trivially also admits a polynomial \OKernel. This can be done using the standard techniques for refuting polynomial kernels or by showing hardness for a parameterized complexity class beyond \FPT, such as e.g.\ \wone, \cowone, \paraNP, or \paracoNP.
    \item The problem is contained in \FPTNP, since otherwise it trivially does not admit a polynomial \OKernel (under standard complexity-theoretical assumptions).
\end{itemize} 
Then, we assert the following:
\begin{itemize}
\item {\bf Satisfiability.}  First, we apply our framework to the canonical \SigmaPtwo-complete problem $\exists\forall$-\textsc{Satisfiability}. We prove that: {\em (i)} $\exists\forall$-DNF parameterized by the number of variables does not admit a polynomial \OKernel unless \NoOKernelAssume; {\em (ii)} $\exists\forall$-DNF admits a polynomial \OKernel when parameterized by the ``size of the existential subformula''. (See \cref{sec:existsforallsat}.)

\item {\bf Clique Hitting.} Second, we apply our framework to \textsc{Clique-Free Vertex Deletion}, which is a natural generalization of \textsc{Vertex Cover} or \textsc{Triangle-Free Vertex Deletion}. The problem known to be \SigmaPtwo-hard and defined as follows. Given graph $G=(V,E)$ and two integers $h$ and $k$, the objective is to determine whether there exists a vertex set $V'\subseteq V$ with $|V'|\le h$ such that $G-V'$ does not contain a clique of size $k$. We prove \textsc{Clique-Free Vertex Deletion} admits a polynomial \OKernel when parameterized by $(k+\#k\text{-cliques})$, where $\#k\text{-cliques}$ is the number of cliques of size $k$ in the input graph. Then, we study the weighted variant of this problem, called \textsc{Weighted Clique-Free Vertex Deletion}. We prove that  \textsc{Weighted Clique-Free Vertex Deletion} still admits a polynomial \OKernel when parameterized by $(k+\#\text{weight-}k\text{-cliques})$, but that \textsc{Weighted Clique-Free Vertex Deletion} parameterized by $(h+k)$ does not admit a polynomial \OKernel unless \NoOKernelAssume.  (See \cref{sec:clique}.)

\item {\bf Discovery Problems.} Third, we study a class of problems termed discovery problems, starting with a concrete example. Specifically, we first consider the \textsc{Discovery Vertex Cover Reconfiguration} problem, defined as follows. The input consists of a vertex set $V$ and for each pair $u,v\in V$ of vertices an instance $I_{u,v}$ of \textsc{Satisfiability}, two minimal vertex covers $S,T\subseteq V$ for the discovered graph $G=(V,E)$, and two integers $k,\ell$. Here, intuitively, we do not know $E$, but can ``discover'' it by solving the instances $I_{u,v}$. Then, the objective is to determine whether there exists a sequence of $\ell$ vertex covers $X_1,\ldots,X_\ell$ of size at most~$k$ for the discovered graph $G$ such that $X_1=S$, $X_\ell=T$, and for each $i\in[\ell-1]$ we have~$|X_i\triangle X_{i+1}|\le 1$.\footnote{We use the symbol $\triangle$ for the symmetric difference. For two sets $S,T$ we have $S\triangle T=(S\cup T)\setminus (S\cap T)$.} We prove that \textsc{Discovery Vertex Cover Reconfiguration} admits a polynomial \OKernel when parameterized by $k$. Afterwards, we consider the ``generic'' form of a so-called discovery problem (that may not necessarily be a graph problem), and provide a meta-theorem for asserting the existence of \OKernels for such problems.
(See \cref{sec:discovery}.)
\end{itemize}

\paragraph{Other Extensions of the Concept of Kernelization.}
Throughout the years, the concept of a kernel has been extended in various ways. 
Since the discussion of these extensions is beyond the scope of this paper, we only briefly mention them. Among them, the oldest and most well-known one is that of a {\em Turing kernel}, where we need not to produce just a single compressed instance of the problem at hand, but are allowed to produce polynomially many such compressed instances. For a discussion of this concept, we refer to Chapter~22 in the book by \citet{Fom+19}. The second most well-known extension of the concept of a kernel is that of a {\em lossy kernel}, which was introduced by \citet{LokshtanovPRS17} and aims to make kernelization suitable for the purpose of computing approximate solutions for optimization problems. For more information, we refer to Chapter~23 in the book by \citet{Fom+19}. In a similar vein to that of a lossy kernel, the concept of a kernel was extended to be suitable also for dealing with other types of problems, such as streaming problems by \citet{fafianie2014streaming}, semi-streaming problems by \citet{LokshtanovMPRSZ24}, dynamic problems by \citet{AlmanMW20}, problems on randomly distributed inputs (average case analysis of kernel sizes) by \citet{friedrich2015kernel}, enumeration problems by \citet{creignou2017paradigms}, and, in several different ways, also for counting problems~(see, e.g., \citet{LokshtanovM0Z24} and citations therein for earlier works). Furthermore, there is research on the complementary question to ours by \citet{BannachT18,BannachT20}: What can we achieve if the kernel is required to be computable in parallel?

\section{Preliminaries}\label{sec:prelims}
In this section, we give an overview of all definitions, notations, and terminology used in this work. For a natural number $n\in \mathbb{N}$, we use $[n]$ to denote the set $\{1,2,\ldots,n\}$.

\paragraph{Classical Complexity.} We use standard concepts of classical computational complexity theory~\cite{AB09,GJ79}. The complexity classes \SigmaPtwo and \PiPtwo~\cite{AB09,Sto76} are located in the second level of the polynomial-time hierarchy and contain both \NP and \coNP.
They are closed under polynomial-time many-one reductions. The class \SigmaPtwo, intuitively, contains all problems that are at most as hard as the problem \textsc{$\exists\forall$-Satisfiability} ($\exists\forall$-DNF)~\cite{AB09,Sto76}, where we are given a Boolean formula $\phi$ in disjunctive normal form and the variables of $\phi$ are partitioned into two sets $X$ and $Y$, and we are asked to decide whether there exists an assignment for all variables in $X$ such that for all possible assignments for the variables in $Y$, such that the formula $\phi$ evaluates to \truevalue. The problem \textsc{$\exists\forall$-DNF} is complete for \SigmaPtwo~\cite{AB09,Sto76}. 
Notably, the problem \textsc{$\exists\forall$-CNF} is contained in \NP~\cite{AB09}. Here, we are given a Boolean formula $\phi$ in conjunctive normal form and the variables of $\phi$ are partitioned into two sets $X$ and~$Y$, and we are asked to decide whether there exists an assignment for all variables in $X$ such that for all possible assignments for the variables in $Y$, such that the formula $\phi$ evaluates to \truevalue. 
The class \PiPtwo, intuitively, contains all problems that are at most as hard as the problem \textsc{$\forall\exists$-Satisfiability} ($\forall\exists$-CNF)~\cite{AB09,Sto76}, where we are given a Boolean formula $\phi$ in conjunctive normal form and the variables of $\phi$ are partitioned into two sets $X$ and $Y$, and we are asked to decide whether for all possible assignments for all variables in $X$ there exists an assignment for the variables in $Y$, such that the formula $\phi$ evaluates to \truevalue. The problem \textsc{$\forall\exists$-CNF} is complete for \PiPtwo~\cite{AB09,Sto76}. The complexity class \PSPACE contains all problems that can be solved using polynomial space and contains both \SigmaPtwo and \PiPtwo~\cite{AB09}.

\paragraph{Parameterized Complexity.} We use standard notation and terminology from parameterized
complexity theory~\cite{Nie06,DF13,FG06,Cyg+15,van2019cognition,Fom+19}
and give here an overview of the concepts used in this work.
A \emph{parameterized problem} is a language $L\subseteq \{0,1\}^* \times \N$. We call the second component
the \emph{parameter} of the problem.
A parameterized problem is \emph{fixed-pa\-ram\-e\-ter tractable} (in the complexity class \FPT{})
if there is an algorithm that solves each instance~$(x,r)\in\{0,1\}\times \mathbb{N}$ in~$f(r) \cdot |x|^{\OO(1)}$ time,
for some computable function $f$.
A decidable parameterized problem $L$ admits a \emph{kernel} if there is a polynomial-time algorithm that transforms each instance $(x,r)\in\{0,1\}\times \mathbb{N}$ into an instance $(x',r')\in\{0,1\}\times \mathbb{N}$ such that $|x'|+ r'\le f(r)$,
for some computable function~$f$, and $(x',r')\in L$ if and only if $(x,r)\in L$. 
It is known that a decidable parameterized problem $L$ admits a kernel if and only if it is fixed-parameter tractable~\cite{cai1997advice}.
A decidable parameterized problem $L$ admits a \emph{polynomial kernel} if there is a polynomial-time algorithm that transforms each instance $(x,r)\in\{0,1\}\times \mathbb{N}$ into an instance $(x',r')\in\{0,1\}\times \mathbb{N}$ such that $|x'|+ r'\in r^{\OO(1)}$ and $(x',r')\in L$ if and only if $(x,r)\in L$. 
Presumably, not all problems that are fixed-parameter tractable admit a polynomial kernel~\cite{Fom+19}.

If a parameterized problem is hard for the parameterized complexity class \wone or \cowone, then it is (presumably) not in~\FPT{}. 
Informally, parameterized problems that are hard for \wone\ are at least as hard as \textsc{Clique} parameterized by the solution size. In \textsc{Clique} we are given a graph $G$ and a ``solution size'' $k$, and are asked to decide whether $G$ contains a set of $k$ vertices that are all pairwise connected by an edge.
The complexity classes \wone\ and \cowone\ are closed under parameterized reductions, which may run in \FPT-time and additionally set the new parameter to a value that exclusively depends on the old parameter. 
If the unparameterized version of a parameterized problem is \NP-hard (resp.\ \coNP-hard) for constant parameter values, then the problem is \paraNP-hard (resp.\ \paracoNP-hard).

\paragraph{Oracle Classes.} 
When we say that an algorithm $A$ or a machine $M$ has \emph{oracle access} to a problem or language $L$, then we assume that $A$ or $M$ can answer queries of the type $x\in L$ in constant time.
The class \PNP is the set of all problems that can be solved by a polynomial-time algorithm that has oracle access to an \NP-complete problem~\cite{AB09,Sto76}.
The class \SigmaPtwo (introduced above) can also be characterized as the set of all problems that can be solved by an \NP-machine that has oracle access to an \NP-complete problem~\cite{AB09,Sto76}.

The class \FPTNP is the set of all parameterized problems that can be solved by an \FPT-algorithm that has oracle access to an \NP-complete problem~\cite{de2017parameterized,de2014fixed,de2019compendium,de2019parameterized,HaanS15}. Natural subclasses are obtained by reducing the number of allowed oracle calls.
Note that restricting the number of oracle calls to be polynomial
in the input size (and independent of the parameter) does not yield a new complexity class.
\begin{lemma}
    Let $C$ be the set of all parameterized problems that can be solved by an \FPT-algorithm that has oracle access to an \NP-complete problem and uses at most $n^{\OO(1)}$ queries to the oracle, where~$n$ is the input size. Then we have that $C$ $=$ \FPTNP.
\end{lemma}
\begin{proof}
    Clearly, we have that $C\subseteq$ \FPTNP. Now let $L$ be a parameterized problem in \FPTNP. Then there exists \FPT-algorithm that has oracle access to an \NP-complete problem and decides $L$. Let $f(r)\cdot n^{\OO(1)}$ denote the running time of the algorithm, where $r$ is the parameter and $n$ is the input size. If $f(r)\le n$ for some instance of $L$, then at most $n^{\OO(1)}$ oracle queries are used to solve that instance. If $f(r)>n$ for some instance of $L$, then that instance can be solved in \FPT-time without oracle queries. It follows that $L\in C$.
\end{proof}



The class \FPTNPfew is the set of all parameterized problems that can be solved by an \FPT-algorithm that has oracle access to an \NP-complete problem and uses at most $f(r)$ queries to the oracle, where $r$ is the parameter and $f$ is some computable function~\cite{de2017parameterized,de2014fixed,de2019compendium,de2019parameterized,HaanS15}.

\paragraph{Advice Classes.} The class \NP/\polyadvice contains all problems that can be solved by an \NP-machine that gets an additional advice string of polynomial length as input~\cite{AB09}. The advice string is allowed to depend on the length of the input, but not on the input itself, and its size needs to be polynomially bounded by the length of the input.
Analogously, the class \SigmaPtwo/\polyadvice contains all problems that can be solved by an \NP-machine with oracle access to an \NP-complete problem that gets an additional advice string of polynomial length as input~\cite{AB09}.

\paragraph{Kernelization Lower Bounds.}
The cross-composition framework~\cite{BJK14,Bod+09,FS11,Dru15,Fom+19,Dell16,DellM14} is a popular tool to refute the existence of a polynomial kernel for a parameterized problem under the assumption that \NNoKernelAssume, the negation of which would cause a collapse of the polynomial-time hierarchy to the third
level. 
Informally, in a cross-composition, we have to \emph{compose} many problem instances of an \NP-hard problem into one big instance of the problem we want to investigate. This composition should then have the property that the big instance is either a \yes-instance if and only if at least one of the input instances is a \yes-instance (in the case of OR-cross-compositions), or if and only if all input instances are \yes-instances (in the case of AND-cross-compositions). This then refutes polynomial kernels for the problem under investigation when parameterized by any parameter that only depends on the maximum size of the input instances (and not on the number of input instances).
In order to formally introduce the framework, we need some definitions first.
An equivalence
relation~$R$ on the instances of some problem~$L$ is a
\emph{polynomial equivalence relation}~if
\begin{enumerate}
 \item one can decide for every two instances in time polynomial in their sizes whether they belong to the same equivalence class, and
 \item for each finite set~$S$ of instances, $R$ partitions the set into at most~$(\max_{x \in S} |x|)^{\OO(1)}$ equivalence classes.  
\end{enumerate}

Using this, we can now define OR-cross-compositions.
An \emph{OR-cross-composition} of a problem~$L\subseteq \{0,1\}^*$ into a
parameterized problem~$P$ (with respect to a polynomial equivalence
relation~$R$ on the instances of~\(L\)) is an algorithm that takes
$t$ $R$-equivalent instances~$x_1,\ldots,x_t$ of~$L$ and
constructs in time polynomial in $\sum_{i=1}^t |x_i|$ an instance
$(y,k)$ of~\(P\) such that
\begin{enumerate}
\item $k$ is polynomially upper-bounded in $\max_{i\in [t]}|x_i|+\log t$, and 
\item $(y,k)$ is a \yes-instance of $P$ if and only if there is an $i\in [t]$ such that $x_{i}$ is a \yes-instance of~$L$. 
\end{enumerate}

If an \NP-hard problem~\(L\) OR-cross-composes into a parameterized
problem~$P$, then~$P$ does not admit a polynomial kernel, unless \NoKernelAssume~\cite{Bod+09,FS11,BJK14,Fom+19,DellM14}.

AND-cross-compositions are defined analogously~\cite{Dru15,BJK14,Dell16,Fom+19}. We omit the discussion of the details here.



\section{Oracle Kernels}\label{sec:orableKernels} 
In this section, we introduce the main concept of this paper: a \OKernel. To this end, we lift the concept of kernels to the oracle setting in the following manner.
\defbox{\OKernel}{
A decidable parameterized problem $L$ admits a \emph{\OKernel} if there is a \PNP-algorithm that transforms each instance $(x,r)\in\{0,1\}^*\times \mathbb{N}$ into an instance $(x', r')\in\{0,1\}^*\times \mathbb{N}$ such that
\begin{enumerate}
    \item $|x'|+ r'\le f(r)$,
for some computable function $f$, and
\item $(x',r')\in L$ if and only if $(x,r)\in L$.
\end{enumerate} 
}
Analogous to the case of kernels and fixed-parameter tractable problems, we get the following.
\begin{theorem}\label{thm:kernelequiv}
    A decidable parameterized problem $L$ is in \FPTNP if and only if it admits a \OKernel.
\end{theorem}
\begin{proof}
Assume that $L$ is in \FPTNP. 
Then there is an \FPT-algorithm that has oracle access to an \NP-complete problem that solves each instance $(x,r)\in\{0,1\}\times \mathbb{N}$ in $f(r)\cdot |x|^{\OO(1)}$ time, for some computable function $f$. We describe a \PNP-kernelization-algorithm in the following. Consider some instance $(x,r)\in\{0,1\}\times \mathbb{N}$. Then, if $f(r)<|x|$, then the instance can be solved by an \PNP-algorithm. Hence, a trivial kernel can be produced by a \PNP-kernelization-algorithm.
If $f(r)>|x|$, then the problem instance is already bounded by a function in the parameter. Hence, the \PNP-kernelization-algorithm can output the original instance.

Assume that the parameterized problem $L$ admits a \OKernel. Consider some instance $(x,r)\in\{0,1\}\times \mathbb{N}$. Then using the \PNP-kernelization-algorithm we can produce an instance $(x',r')\in\{0,1\}\times \mathbb{N}$ with $|(x',r')|\le f(r)$, for some computable function $f$. Since $L$ is decidable, there is an algorithm $A$ solving $(x',r')$ with running time $g(|(x',r')|)$, for some computable function $g$. Hence, we can solve $(x,r)$ in $g(f(r))\cdot |x|^{\OO(1)}$ time with oracle access to an \NP-complete problem by first invoking the \PNP-kernelization-algorithm and then algorithm $A$.
\end{proof}

Similar to the classical setting, we are mostly interested in finding \OKernel{s} of polynomial size. To this end, we give the main definition of this work in the following.

 \defbox{Polynomial \OKernel}{
A decidable parameterized problem $L$ admits a \emph{polynomial \OKernel} if there is a \PNP-algorithm that transforms each instance $(x,r)\in\{0,1\}^*\times \mathbb{N}$ into an instance $(x', r')\in\{0,1\}^*\times \mathbb{N}$ such that
\begin{enumerate}
    \item $|x'|+r'\in r^{\OO(1)}$, and
\item $(x',r')\in L$ if and only if $(x,r)\in L$.
\end{enumerate}} 

We further give a definition of (polynomial) compressions in the oracle setting, which is analogous to the ``classical'' (polynomial) compression~\cite{Fom+19}. A decidable parameterized problem $L$ admits a \emph{(resp.\ polynomial) \PNP-compression} into a problem $R\subseteq\{0,1\}^*$ if there is a \PNP-algorithm that transforms each instance $(x,r)\in\{0,1\}^*\times\mathbb{N}$ into a string $y$ such that 
\begin{enumerate}
    \item $|y|\le f(r)$,
for some computable function $f$, (resp.\ $|y|\in r^{\OO(1)}$), and
    \item $y \in R$ if and only if $(x,r)\in L$.
\end{enumerate}
Note that (polynomial) \PNP-compressions are weaker than (polynomial) \OKernel{s} in the following sense.

\begin{observation}\label{obs:kernelcompression}
    If a parameterized problem $L$ admits a (polynomial) \OKernel, then it admits a (polynomial) \PNP-compression into the unparameterized version of itself.
\end{observation}

We remark that \citet[Chapter 16]{de2019parameterized} discusses ``non-deterministic kernelization''. He introduces a kernelization concept where the kernelization algorithm may run in non-deterministic polynomial time. This kernelization concept captures the parameterized complexity class \paraNP, that is, a decidable parameterized problem admits a kernelization algorithm of the described kind if and only if it is contained in \paraNP. For technical details, we refer to \citet[Chapter 16]{de2019parameterized}.



\section{Oracle Kernelization Lower Bounds}\label{sec:lower}
In this section, we present a method for refuting the existence of polynomial \OKernel{s} for parameterized problems.
To this end, we lift the OR-cross-composition framework~\cite{Bod+09,FS11,BJK14,Fom+19,DellM14} discussed in \cref{sec:prelims} to the oracle setting. In this setting, we can make the polynomial equivalence relation stronger, as we can also employ oracle accesses here. An equivalence
relation~$R$ on the instances of some problem~$L$ is a
\emph{polynomial \PNP-equivalence relation} if
\begin{enumerate}
 \item one can decide for every two instances in \PNP-time in their sizes whether they belong to the same equivalence class, and
 \item for each finite set~$S$ of instances, $R$ partitions the set into at most~$(\max_{x \in S} |x|)^{\OO(1)}$ equivalence classes.  
\end{enumerate}

Using this, we can now define \PNP-OR-cross-compositions. 

\defbox{\PNP-OR-cross-composition}{
A \emph{\PNP-OR-cross-composition} of a problem~$L\subseteq \{0,1\}^*$ into a
parameterized problem~$P$ (with respect to a polynomial \PNP-equivalence
relation~$R$ on the instances of~\(L\)) is an algorithm that takes
$t$ $R$-equivalent instances~$x_1,\ldots,x_t$ of~$L$ and
constructs in \PNP-time in $\sum_{i=1}^t |x_i|$ an instance
$(y,k)$ of~\(P\) such that
\begin{enumerate}
\item $k$ is polynomially upper-bounded in $\max_{i\in[t]}|x_i|+\log t$, and 
\item $(y,k)$ is a \yes-instance of $P$ if and only if there is an $i\in [t]$ such that $x_{i}$ is a \yes-instance of~$L$. 
\end{enumerate}}

Formally, we can use \PNP-OR-cross-compositions to refute the existence of polynomial \OKernel{s} as follows.

\begin{theorem}\label{thm:okernellowerbound}
   If a \SigmaPtwo-hard problem~\(L\) \PNP-OR-cross-composes into a parameterized
problem~$P$, then~$P$ does not admit a polynomial \OKernel, unless \NoOKernelAssume. 
\end{theorem}

We prove \cref{thm:okernellowerbound} by adapting the proof techniques used for the classical OR-cross-composition result~\cite{Bod+09,FS11,BJK14,Fom+19,DellM14} (see \cref{sec:prelims}) in a straightforward way. We follow the presentation of the proof by \citet{Fom+19}. To this end, we introduce some additional concepts, particularly the one of a \PNP-OR-distillation. Let $L,R\subseteq \{0,1\}^*$ be two problems and let $t:\mathbb{N}\rightarrow\mathbb{N}$ be some computable function. Then a \emph{$t$-bounded \PNP-OR-distillation} is an algorithm that for every $n$, given $t(n)$ strings $x_1,\ldots,x_{t(n)}$ of length $n$ as input,
\begin{itemize}
    \item runs in \PNP-time, and 
    \item outputs a string $y$ of length at most $t(n)\cdot \log n$ such that $y\in R$ if and only if $x_i\in L$ for some $i\in[t(n)]$.
\end{itemize}
We first show that the existence of a polynomially bounded \PNP-OR-distillation for an \SigmaPtwo-hard problem implies \NoOKernelAssume.
\begin{proposition}\label{prop:dist}
    If there is a $t$-bounded \PNP-OR-distillation algorithm $A$ from a problem $L\subseteq \{0,1\}^*$ into a problem $R\subseteq \{0,1\}^*$ for some polynomially bounded function $t$, then $\Bar{L}\in \mathsf{\Sigma_2^\text{P}}/\textsf{poly}$.\footnote{We use $\Bar{L}$ to denote the complement of $L$, that is, $\Bar{L}=\{x\in\{0,1\}^*\mid x\notin L\}$.} In particular, if $L$ is \SigmaPtwo-hard, then \NoOKernelAssume.
\end{proposition}

The following known lemma is the main tool to construct the advice strings for the proof of \cref{prop:dist}.
\begin{lemma}[\cite{Fom+19}]\label{lem:aux}
    Let $X,Y$ be finite sets, $p$ a natural number, and $\beta:X^p\rightarrow Y$ be a mapping. We say that $y\in Y$ covers $x\in X$ if there exist $x_1,\ldots,x_p\in X$ such that $x_i=x$ for some $i\in[p]$, and $\beta((x_1,\ldots,x_p))=y$. Then at least one element from $Y$ covers at least $|X|/\sqrt[p]{|Y|}$ elements of $X$.
\end{lemma}

Now we prove \cref{prop:dist}.

\begin{proof}[Proof of \cref{prop:dist}]
    For each $n\in\mathbb{N}$, let $\Bar{L}_n=\{x\in\Bar{L}\mid |x|=n\}$ and let $\alpha(n)=t(n)\cdot\log n$. Then for each input $x_1,\ldots,x_{t(n)}$ such that for all $i\in[t(n)]$ we have $x_i\in\Bar{L}_n$, algorithm $A$ outputs a string $y$ of length at most $\alpha(n)$ such that $y\in\Bar{R}$. Let $\Bar{R}_{\alpha(n)}=\{y\in\Bar{R}\mid |y|\le\alpha(n)\}$. Thus we have that algorithm $A$ maps each $t(n)$-tuple $x_1,\ldots,x_{t(n)}$ with $x_i\in\Bar{L}_n$ for all $i\in[t(n)]$ to some $y\in\Bar{R}_{\alpha(n)}$.
    It is known that by repeated application of \cref{lem:aux}, we can construct a set $S_n\subseteq \Bar{R}_{\alpha(n)}$ of size $t(n)^{\OO(1)}$ that covers every element of $\Bar{L}_n$~\cite{Fom+19}. For every string $x$ of length $n$, by construction of $S_n$ we have the following.
    \begin{itemize}
        \item If $x\in\Bar{L}_n$, then there is a $t(n)$-tuple $x_1,\ldots,x_{t(n)}$ of strings of length $n$ with $x_i=x$ for some $i\in[t(n)]$, such that on input $x_1,\ldots,x_{t(n)}$ the algorithm $A$ outputs some $y\in S_n$.
        \item If $x\notin \Bar{L}_n$, then by construction we have that $x\in L$. Then for every $t(n)$-tuple $x_1,\ldots,x_{t(n)}$ of strings of length $n$ with $x_i=x$ for some $i\in[t(n)]$, the algorithm $A$ outputs a string from $R$ on input $x_1,\ldots,x_{t(n)}$, and thus not a string from $S_n\subseteq\Bar{R}$. 
    \end{itemize}
    Now we construct a \SigmaPtwo/\polyadvice machine $M$ that decides the language $\Bar{L}$ as follows. For $x\in\{0,1\}^n$ as input, the machine $M$ takes the set $S_n$ as advice. First, the machine non-deterministically guesses a $t(n)$-tuple $x_1,\ldots,x_{t(n)}$ of strings of length $n$ with $x_i=x$ for some $i\in[t(n)]$. Then $M$ runs algorithm $A$ (in \PNP-time) on input $x_1,\ldots,x_{t(n)}$. If $A$ outputs a string from $S_n$, then $M$ accepts. Otherwise, $M$ rejects. This completes the description of $M$.

    It is clear that $M$ is a \SigmaPtwo/\polyadvice-machine, since it makes polynomially many non-deterministic guesses and runs in \PNP-time. By the above two properties, we have that $x\in\Bar{L}$ if and only if $M$ accepts $x$, that is, there is a computation path that accepts $x$. Thus we have $\Bar{L}\in\mathsf{\Sigma_2^\text{P}}/\textsf{poly}$. Furthermore, if $L$ is \SigmaPtwo-hard, then $\Bar{L}$ is \PiPtwo-hard and every problem in \PiPtwo is reducible to $\Bar{L}$. Hence, we get \NoOKernelAssume. This completes the proof.
\end{proof}

Now we are ready to prove the main result of this section (\cref{thm:okernellowerbound}).

\begin{proof}[Proof of \cref{thm:okernellowerbound}]
    We prove a stronger statement, namely that if a \SigmaPtwo-hard problem~\(L\) \PNP-OR-cross-composes into a parameterized
problem~$P$, then~$P$ does not admit a polynomial \PNP-compression, unless \NoOKernelAssume. By \cref{obs:kernelcompression} this yields the result.

Suppose that $P$ has a polynomial \PNP-compression into some problem $Q$. We will construct a $t$-bounded \PNP-OR-distillation algorithm $A$ with polynomially bounded $t$ from $L$ into $Q^{\text{OR}}$, where $Q^{\text{OR}}$ is defined as follows. For each $n\in\mathbb{N}$, let $p_n:(\{0,1\}^*)^n\rightarrow \{0,1\}^*$ be a bijective function that combines $n$ bitstrings into one string such that $|p_n(x_1,\ldots,x_n)|=n+2\sum_{i=1}^n|x_i|$.\footnote{Note that this can be achieved e.g.\ by putting a zero in front of each bit of each string and using one as a separator symbol.} For each $n\in\mathbb{N}$ and each $n$-tuple $x_1,\ldots,x_n$ of strings we have that $p_n(x_1,\ldots,x_n)\in Q^{\text{OR}}$ if and only if $x_i\in Q$ for some $i\in[n]$. This together with \cref{prop:dist} will prove the theorem.

We start with some straightforward observations.
\begin{enumerate}
    \item Problem~\(L\) \PNP-OR-cross-composes into parameterized problem~$P$. Let $R$ be the polynomial \PNP-equivalence relation. There exists some $c_1\in \mathbb{N}$ such that given a set of $R$-equivalent instances $x_1,\ldots,x_t$ of $L$ the \PNP-OR-cross-composition produces an instance $(y,k)$ of $P$ with~$k\le (\max_{i\in[t]}|x_i|+\log(t))^{c_1}$.

    \item Problem $P$ has a polynomial \PNP-compression into problem $Q$. There exists some $c_2\in \mathbb{N}$ such that given an instance $(x,k)$ of $P$ the polynomial \PNP-compression produces an instance $y$ of $Q$ with~$|y|\le k^{c_2}$.
    \item Since $R$ is a polynomial \PNP-equivalence relation, there exists some $c_3\in \mathbb{N}$ such that every set of strings of length $n$ is partitioned into at most $n^{c_3}$ equivalence classes.
\end{enumerate}
We now construct a $t$-bounded \PNP-OR-distillation algorithm $A$ with $t(n)=2n^{2(c_1\cdot c_2+c_3)}$ from $L$ into~$Q^{\text{OR}}$. Assume we are given a $t(n)$-tuple $x_1,\ldots,x_{t(n)}$ of strings of length $n$ that are instances of~$L$. We partition the instances in \PNP-time into $R$-equivalence classes $X_1,\ldots,X_r$ with $r\le n^{c_3}$. We apply the \PNP-OR-cross-composition algorithm to each class $X_i$ with $i\in[r]$ to obtain instances $(y_i,k_i)$ with $i\in[r]$ of $P$. Note that we have $k_i\le (n+\log(t(n)))^{c_1}$ for each $i\in[r]$. 

Now we use the polynomial \PNP-compression from $P$ into problem $Q$ to produce instances $z_1,\ldots,z_r$ of $Q$ with $|z_i|\le k_i^{c_2}$ for each $i\in[r]$. Observe that $z=p_r(z_1,\ldots,z_r)$ is an instance of~$Q^{\text{OR}}$. The length of $z$ is at most $r+2\sum_{i=1}^r|z_i|$ which is at most $n^{c_3}+2(n+\log(t(n)))^{c_1\cdot c_2}$. From here it is straightforward to verify that $|z|\le t(n)$. Thus we have constructed a $t$-bounded \PNP-OR-distillation algorithm $A$ from $L$ into~$Q^{\text{OR}}$ with polynomially bounded $t$. This together with \cref{prop:dist} completes the proof.
\end{proof}

We conjecture that analogously defined \PNP-AND-cross-compositions can also be used to refute polynomial \OKernel{s} based on the same assumptions, but leave this for future research. We expect this to be much harder to prove, similarly to the classical AND-cross-compositions~\cite{Dru15,Dell16}.


\section{Application to {\boldmath $\exists\forall$}-\textsc{Satisfiability}}
\label{sec:existsforallsat}

In this section, we apply our framework to the canonical \SigmaPtwo-complete problem $\exists\forall$-\textsc{Satisfiability}, which is defined as follows.

\problemdef{{\boldmath $\exists\forall$}-\textsc{Satisfiability} ({\boldmath $\exists\forall$}-DNF)}{A Boolean formula $\phi$ in disjunctive normal form, where the variables of $\phi$ are partitioned into two sets $X$ and $Y$.}{Does there exist an assignment for the variables in $X$ such that for all assignments for the variables in $Y$, the formula $\phi$ evaluates to \truevalue?}

\paragraph{Parameterizations.} This problem has several natural parameters that we will consider: 
\begin{itemize}
    \item the number $|X|$ of existentially quantified variables, 
    \item the number $|Y|$ of universally quantified variables, 
    \item the total number of variables $|X|+|Y|$, and
    \item the ``size of existential subformula''.
\end{itemize}
 The latter intuitively measures the size of the subformula that involves existentially quantified variables. Formally, given an instance~$\phi$ of $\exists\forall$-DNF, let $C$ be the set of all clauses that contain existentially quantified variables. Now exhaustively add to $C$ all clauses that only contain universally quantified variables where at least one of those variables also appears in another clause in $C$. Let~$\Bar{C}$ be the remaining clauses. Clearly, the clauses in $\Bar{C}$ only contain universally quantified variables, and no variable that appears in a clause in $\Bar{C}$ also appears in a clause in $C$. Now let $\phi_1$ be the subformula of $\phi$ composed of all clauses in $C$ and let $\phi_2$ be the subformula composed of all clauses in $\Bar{C}$. W.l.o.g.\ we have that $\phi=\phi_1\vee\phi_2$. Now we define the \emph{size of existential subformula} as $|\phi_1|=\sum_{c\in C}|c|$, where for some clause $c$ we denote by $|c|$ the number of variables appearing in $c$. Consider the following example:
 \[
 \phi=(x_1\wedge \neg x_2 \wedge y_1) \vee (\neg y_1 \wedge y_2) \vee (y_3 \wedge \neg y_4), 
 \]
 where $X=\{x_1,x_2\}$ and $Y=\{y_1,y_2,y_3,y_4\}$.
 Here, we first add $(x_1\wedge \neg x_2 \wedge y_1)$ to $C$ since it contains existentially quantified variables. Then we add $(\neg y_1 \wedge y_2)$ to $C$ since it only contains universally quantified variables where at least one of those variables also appears in another clause in $C$, in this case $y_1$. The remaining clause $(y_3 \wedge \neg y_4)$ also contains universally quantified variables but none of them also appears in another clause in $C$, hence this clause is in $\Bar{C}$. For this example, we get $\phi_1=(x_1\wedge \neg x_2 \wedge y_1) \vee (\neg y_1 \wedge y_2)$ and $\phi_2=(y_3 \wedge \neg y_4)$.
 Note that the size of the existential subformula is always at least the number of existentially quantified variables, that is, $|X|\le|\phi_1|$. 
 
 For the mentioned parameterizations, we can straightforwardly obtain the following results.

\begin{theorem}\label{thm:existforallhardness}
    The problem $\exists\forall$-DNF is
    \begin{itemize}
        \item \SigmaPtwo-hard,
        \item \paracoNP-hard when parameterized by the size of the existential subformula, and
        \item in \FPTNPfew when parameterized by the number of existentially quantified variables.
    \end{itemize}
\end{theorem}
\begin{proof}
    The first statement of the theorem is a well-known fact~\cite{AB09,Sto76}. The second statement follows from the observation that if the size of the existential subformula is zero, the problem is equivalent to \textsc{Tautology}\footnote{In \textsc{Tautology}, we are given a Boolean formula in disjunctive normal form, and are asked whether this formula evaluates to \truevalue for all assignments for the variables.}, which is a well-known \coNP-complete problem~\cite{Karp72}. 

    For the third statement, we give the following straightforward algorithm. We enumerate all assignments for the existentially quantified variables. For each one, we first substitute the existentially quantified variables in $\phi$ with the respective truth values of the assignment. Note that the remaining formula is a \textsc{Tautology} instance. We use the \NP-oracle to check whether the remaining \textsc{Tautology} instance is a \yes-instance.\footnote{Note that since oracles can also decide complement languages, we can use the \NP-oracle to solve \coNP-complete problems.} If we find an assignment for the existentially quantified variables such that the corresponding \textsc{Tautology} instance is a \yes-instance, we answer \yes. Otherwise, we answer \no. It is straightforward to verify that the described algorithm is correct and runs in \FPTNPfew-time.
\end{proof}
Notably, we leave the parameterized complexity of $\exists\forall$-DNF with respect to the number of universally quantified variables open.

\paragraph{Kernelization.} From \cref{thm:existforallhardness} the natural question arises whether $\exists\forall$-DNF parameterized by the number of existentially quantified variables admits a polynomial \OKernel. In the following, we answer this question negatively.

\begin{theorem}\label{thm:satnokernel}
    The problem $\exists\forall$-DNF parameterized by the number of variables does not admit a polynomial \OKernel unless \NoOKernelAssume.
\end{theorem}
\begin{proof}
    To prove the theorem, we give a \PNP-OR-cross-composition from $\exists\forall$-DNF (unparameterized) into $\exists\forall$-DNF parameterized by the number of variables. Since $\exists\forall$-DNF is \SigmaPtwo-hard~\cite{AB09,Sto76},  by \cref{thm:okernellowerbound} this yields the result. 

    First of all, we consider two instances of $\exists\forall$-DNF $R$-equivalent if they have the same number of existentially quantified variables and the same number of universally quantified variables. Clearly, we have that $R$ is a polynomial \PNP-equivalence relation.

    Now assume we are given $t$ $R$-equivalent instances $\phi_1,\ldots,\phi_t$ of $\exists\forall$-DNF that each have $n_1$ existentially quantified variables and $n_2$ universally quantified variables. Let $X_i$ denote the set of existentially quantified variables of $\phi_i$ and let $Y_i$ denote the set of universally quantified variables of $\phi_i$, for all $i\in[t]$. Assume w.l.o.g.\ that $t$ is a power of two. Assume all existentially quantified variables are ordered in some fixed but arbitrary way. We identify the $j$th existentially quantified variable of each instance $\phi_i$ with $x_j$ for all $i\in[t]$ and $j\in[n_1]$. Analogously, assume all universally quantified variables are ordered in some fixed but arbitrary way. We identify the $j$th universally quantified variable of each instance $\phi_i$ with $y_j$ for all $i\in[t]$ and $j\in[n_2]$. Hence, from now on we assume that all instances have the same set $X$ of existentially quantified variables and the same set~$Y$ of universally quantified variables.
    
    We create an instance $\phi^\star$ as follows. First, we add $\log t$ fresh variables $z_1,\ldots,z_{\log t}$ to the set $X$ of existentially quantified variables. We call those the ``instance selection variables''. 
    For each $i\in[t]$ we create an instance $\phi_i'$ from $\phi_i$ by adding $\ell_1\wedge\ldots\wedge\ell_{\log t}$ to each clause of $\phi_i$, where $\ell_j=z_j$ if the $j$th bit of the binary representation of $i$ is one, and $\ell_j= \neg z_j$ otherwise. Here, we start counting at the least significant bit.
    Intuitively, each assignment for the instance selection variables deactivates all but one instance. We give an example in \cref{fig:satexample}.
    Now we set $\phi^\star=\phi'_1\vee\ldots\vee\phi'_t$. The instance~$\phi^\star$ can clearly be computed in \PNP-time.

\begin{figure}[t]
    \centering
    \[
    \phi_1 = a_1\vee a_2, \ \ \ \ \phi_2 = b_1\vee b_2, \ \ \ \ \phi_3 = c_1\vee c_2, \ \ \ \ \phi_4 = d_1\vee d_2
    \]

\vspace{-3ex}
    
    \[
    \phi'_1 = (a_1\wedge z_1 \wedge \neg z_2)\vee (a_2\wedge z_1 \wedge \neg z_2), \ \ \ \ \phi'_2 = (b_1\wedge \neg z_1 \wedge z_2)\vee (b_2\wedge \neg z_1 \wedge z_2),
    \]
    \[
    \phi'_3 = (c_1\wedge z_1 \wedge z_2)\vee (c_2\wedge z_1 \wedge z_2), \ \ \ \ \phi'_4 = (d_1\wedge \neg z_1 \wedge \neg z_2)\vee (d_2\wedge \neg z_1 \wedge \neg z_2)
    \]

\vspace{-3ex}

    \[
\phi^\star=\phi'_1\vee\phi'_2\vee\phi'_3\vee\phi'_4
    \]
    \caption{Example of the \PNP-OR-cross-composition presented in the proof of \cref{thm:satnokernel}. The top row shows four example input instances. Here, $a_1$ and $a_2$ are two clauses of $\phi_1$. Analogously, $b_1,b_2$ and $c_1,c_2$ and $d_1,d_2$ are two clauses of $\phi_2$ and $\phi_3$ and $\phi_4$, respectively. The middle two rows show how the input instances are modified by adding two extra literals involving variables $z_1,z_2$ to each clause. The bottom row shows the output instance.}
    \label{fig:satexample}
\end{figure}

    First, note that $\phi^\star$ is an instance of $\exists\forall$-DNF and the number of variables of $\phi^\star$ is $n=n_1+n_2+\log t$, which is clearly polynomially upper-bounded by $\max_{u\in[t]}|\phi_i|+\log t$. In the remainder, we show that $(\phi^\star,n)$ is a \yes-instance of $\exists\forall$-DNF parameterized by the number of variables if and only if there is some $i\in[t]$ such that $\phi_i$ is a \yes-instance of $\exists\forall$-DNF.

    $(\Rightarrow)$: Assume that $(\phi^\star,n)$ is a \yes-instance of $\exists\forall$-DNF parameterized by the number of variables. Then there is an assignment for the existentially quantified variables of $\phi^\star$ such that for all assignments of the universally quantified variables of $\phi^\star$, the formula $\phi^\star$ evaluates to \truevalue. Assume we have such an assignment for the existentially quantified variables. Consider the assignment for the instance selection variables $z_1,\ldots,z_{\log t}$ (which are existentially quantified). Let the binary representation of $i$ correspond to this assignment, that is, the $j$th bit of the binary representation of $i$ is one if and only if $z_j$ is set to \truevalue. Now consider some $i'$ and the subformula $\phi'_{i'}$ of $\phi^\star$. The formula $\phi'_{i'}$ was created from $\phi_i$ by adding $\ell_1\wedge\ldots\wedge\ell_{\log t}$ to each clause of $\phi_{i'}$, where $\ell_j=z_j$ if the $j$th bit of the binary representation of $i'$ is one, and $\ell_j= \neg z_j$ otherwise. We have that $\ell_1\wedge\ldots\wedge\ell_{\log t}$ evaluates to \truevalue if and only if $i'=i$. In other words, each clause of each subformula $\phi'_{i'}$ with $i'\neq i$ evaluates to \falsevalue for each assignment of the universally quantified variables. It follows that for each assignment of the universally quantified variables, at least one clause of $\phi_i'$ evaluates to \truevalue. This implies that $\phi_i$ is a \yes-instance of $\exists\forall$-DNF.

    $(\Leftarrow)$: Assume that $\phi_i$ with $i\in[t]$ is a \yes-instance of $\exists\forall$-DNF. Then there is an assignment for the existentially quantified variables of $\phi_i$ such that for all assignments of the universally quantified variables of $\phi_i$, the formula $\phi_i$ evaluates to \truevalue. Assume we have such an assignment for the existentially quantified variables. We construct an assignment for the existentially quantified variables of $\phi^\star$ as follows. For all existentially quantified variables that $\phi^\star$ has in common with $\phi_i$ we use the same assignment as for $\phi_i$. Note that the only existentially quantified variables that $\phi^\star$ does not have in common with $\phi_i$ are the instance selection variables $z_1,\ldots,z_{\log t}$. For each $j\in[\log t]$ we set $z_j$ to \truevalue if and only if the $j$th bit of the binary representation of $i$ is one. Note that then for each clause of the subformula $\phi_i'$ of $\phi^\star$ we have that it evaluates to \truevalue for some assignment for the universally quantified variables if and only if the corresponding clause in $\phi_i$ evaluates to \truevalue for the same assignment for the universally quantified variables. It follows that for each assignment for the universally quantified variables, at least one clause of $\phi_i'$ evaluates to \truevalue. This implies that that $(\phi^\star,n)$ is a \yes-instance of $\exists\forall$-DNF parameterized by the number of variables.
\end{proof}

However, as we show next, if we parameterize $\exists\forall$-DNF by the size of the existential subformula, which is a larger parameter than the number of existentially quantified variables and incomparable to the total number of variables, then we can straightforwardly obtain a polynomial \OKernel.

\begin{theorem}
    The problem $\exists\forall$-DNF admits a polynomial \OKernel when parameterized by the size of the existential subformula.
\end{theorem}
\begin{proof}
     Let $(\phi^\star,k)$ be an instance of $\exists\forall$-DNF parameterized by the size of the existential subformula. W.l.o.g.\ let $\phi^\star=\phi_1\vee\phi_2$, where $|\phi_1|=k$, and $\phi_2$ does not contain any existentially quantified variables and every variable that appears in $\phi_2$ does not appear in $\phi_1$. Note that we can compute~$\phi_1$ and $\phi_2$ in a straightforward manner from $\phi^\star$ in polynomial time.

     Now we use the following \PNP-algorithm to produce an equivalent instance of polynomial size in $k$. We check (using the \NP-oracle) whether $\phi_2$ is a \yes-instance of \textsc{Tautology}. If it is, then we output a trivial \yes-instance of $\exists\forall$-DNF that has constant size. Otherwise, we output $\phi_1$. Clearly, the instance that is produced by the above-described algorithm has a size polynomial in $k$. It is straightforward to verify that the produced instance is equivalent to $\phi^\star$.
\end{proof}

\paragraph{Open Questions.} As already mentioned, we leave the parameterized complexity of $\exists\forall$-DNF with respect to the number of universally quantified variables open. We conjecture that we can obtain similar results for the \PiPtwo-complete problem $\forall\exists$-CNF. However, to refute polynomial \OKernel for $\forall\exists$-CNF parameterized by the number of universally quantified variables, we would (presumably) need a \PNP-AND-cross-composition.

\section{Application to \textsc{Clique-Free Vertex Deletion}}\label{sec:clique}
In this section, we apply our framework to \textsc{Clique-Free Vertex Deletion}, which is a natural generalization of \textsc{Vertex Cover} or \textsc{Triangle-Free Vertex Deletion}. The problem is known to be \SigmaPtwo-hard and defined as follows.

\problemdef{\textsc{Clique-Free Vertex Deletion}}{A graph $G=(V,E)$ and two integers $h$ and $k$.}{Does there exist a vertex set $V'\subseteq V$ with $|V'|\le h$ such that $G-V'$ does not contain a clique of size $k$?}

\paragraph{Parameterizations.}

 This problem has several natural parameters that we will consider:
\begin{itemize}
    \item the size $h$ of the deletion set,
    \item the size $k$ of the cliques that are to be deleted, and
    \item the number \#$k\text{-cliques}$ of cliques of size $k$ in $G$.
\end{itemize}
Note the latter can be exponentially larger than the input size, however, we will argue that the problem remains hard even if the number of $k$-cliques in the input graph $G$ is small. Note that we can assume w.l.o.g.\ that $h\le \#k\text{-cliques}$, since if $h>\#k\text{-cliques}$ we obtain a trivial \yes-instance: we can delete an arbitrary vertex from every clique of size $k$.
 For the mentioned parameterizations, we establish the following results.

\begin{theorem}\label{thm:cliquefree}
    The problem \textsc{Clique-Free Vertex Deletion} is 
    \begin{itemize}
        \item \SigmaPtwo-hard,
        \item \cowone-hard when parameterized by $k$ even if~$h=0$, 
        \item \cowone-hard (under randomized reductions) when parameterized by $k$ even if \#$k\text{-cliques}=1$, and
        \item in \FPTNPfew when parameterized by $(h+k)$.
    \end{itemize}
\end{theorem}

\begin{proof}
The first statement of the theorem is a known result~\cite{rutenburg1994propositional}.
The second statement from \cref{thm:cliquefree} follows from the well-known fact that \textsc{Clique} parameterized by the solution size is \wone-hard~\cite{DowneyF95,DF13}. The third result from \cref{thm:cliquefree} follows from the fact that \textsc{Clique} parameterized by the solution size remains \wone-hard under randomized reductions even if the input instances are restricted to graphs $G$ that either contain exactly one clique of size $k$ or none~\cite{DF13,MontoyaM13}.\footnote{More precisely, \citet{MontoyaM13} showed that the \textsc{Clique} parameterized by the solution size $k$ with the ``promise'' that the input instance contains at most one clique of size $k$ is \wone-hard under randomized reductions.} 

In the remainder, we show the fourth result from \cref{thm:cliquefree}.
%
    To prove that \textsc{Clique-Free Vertex Deletion} is in \FPTNPfew when parameterized by $(h+k)$ we present a simple search-tree algorithm. The algorithm works as follows.
    \begin{itemize}
        \item Find a clique of size $k$ using the \NP-oracle. If no clique is found, answer \yes.
        \item Branch on which vertex of the clique to delete and decrease $h$ by one. If $h$ becomes negative, abort the current branch of the search-tree.
        \item If no branch of the search-tree yields the answer \yes, then answer \no.
    \end{itemize}
    The search-tree has depth $h$ since we can delete at most $h$ vertices. The branching-width is $k$ since we have $k$ options on which vertex to delete. Hence, the search-tree has size $\OO(k^h)$. In each node of the search-tree, we invoke the \NP-oracle once. Hence, the described algorithm runs in \FPTNPfew-time for parameter $(h+k)$.

    If the algorithm answers \yes, then we have deleted at most $h$ vertices such that the remaining graph does not contain a clique of size $k$. Hence, we face a \yes-instance. If we face a \yes-instance, then there is a solution that deletes at least one vertex from every clique of size $k$. Since the described algorithm exhaustively tries out all possibilities to delete vertices from cliques of size $k$, one of the branches of the search-tree leads to the solution and causes the algorithm to answer \yes.
\end{proof}

\paragraph{Kernelization.} \cref{thm:cliquefree} raises the natural question of whether \textsc{Clique-Free Vertex Deletion} admits a polynomial \OKernel when parameterized by $(k+h)$. While leaving this question open, we show that we can obtain a polynomial \OKernel for a larger parameterization.

\begin{theorem}
    The problem \textsc{Clique-Free Vertex Deletion} admits a polynomial \OKernel when parameterized by $(k+\#k\text{-cliques})$.
\end{theorem}
\begin{proof}
    To prove that \textsc{Clique-Free Vertex Deletion} parameterized by $(k+\#k\text{-cliques})$ admits a polynomial \OKernel we present the following reduction rule.
    \begin{itemize}
        \item If a vertex is not contained in any clique of size $k$, remove that vertex.
    \end{itemize}
    This reduction rule is clearly safe. Furthermore, it can be applied exhaustively in \PNP-time: We can use the \NP-oracle to check for some vertex $v$ whether this vertex is contained in a clique of size $k$ or not, since this problem is clearly contained in \NP. We check this for every vertex at most once. If it cannot be applied anymore, each vertex is included in at least one clique of size $k$. It follows that the overall number of vertices is at most $k\cdot\#k\text{-cliques}$. Hence, the algorithm created an equivalent instance with a size polynomial in $k+\#k\text{-cliques}$.
\end{proof}

We can obtain additional results for the following weighted version of \textsc{Clique-Free Vertex Deletion}.

\problemdef{\textsc{Weighted Clique-Free Vertex Deletion}}{A graph $G=(V,E)$ with (unarily encoded) vertex weights $w:V\rightarrow \mathbb{N}$ and two integers $h$ and $k$.}{Does there exists a vertex set $V'\subseteq V$ with $\sum_{v\in V'}w(v)\le h$ such that $G-V'$ does not contain a clique of weight $k$?}

Note that the proofs of our positive results for \textsc{Clique-Free Vertex Deletion} straightforwardly carry over to \textsc{Weighted Clique-Free Vertex Deletion}. Hence, we have the following.

\begin{proposition}
    The problem \textsc{Weighted Clique-Free Vertex Deletion} is in \FPTNPfew when parameterized by $(h+k)$, and admits a polynomial \OKernel when parameterized by $(k+\#\text{weight-}k\text{-cliques})$.
\end{proposition}

For \textsc{Weighted Clique-Free Vertex Deletion}, we can indeed rule out a polynomial \OKernel for the parameterization $(h+k)$.

\begin{theorem}\label{thm:nopk2}
    The problem \textsc{Weighted Clique-Free Vertex Deletion} parameterized by $(h+k)$ does not admit a polynomial \OKernel unless \NoOKernelAssume.
\end{theorem}
\begin{proof}
     To prove the theorem, we give a \PNP-OR-cross-composition from \textsc{Clique-Free Vertex Deletion} (unparameterized) into \textsc{Weighted Clique-Free Vertex Deletion} parameterized by $(h+k)$. Since \textsc{Clique-Free Vertex Deletion} is \SigmaPtwo-hard~\cite{rutenburg1994propositional}, by \cref{thm:okernellowerbound} this yields the result. 

     First of all, we consider two instances of \textsc{Clique-Free Vertex Deletion} $R$-equivalent if they have the same number of vertices and the same values for $h$ and $k$. Clearly, we have that $R$ is a polynomial \PNP-equivalence relation.

    Now assume we are given $t$ $R$-equivalent instances $x_1,\ldots,x_t$ of \textsc{Clique-Free Vertex Deletion} that each have $n$ vertices. Let $G_1, \ldots, G_t$ be the graphs of the respective instances. Assume w.l.o.g.\ that $t$ is a power of two.
    
    We create an instance $(G^\star,w,h^\star,k^\star)$ as follows. We start by setting $h^\star=h+n\cdot \log t$ and $k^\star=k+n\cdot \log t$, where we assume w.l.o.g.\ that $n>h,k>0$. We set $G^\star$ to be the disjoint union of $G_1,\ldots,G_t$. We set the weight of all $n\cdot t$ vertices of $G^\star$ to one. Now we add $2\log t$ ``selection vertices'' $v_1,v_1',v_2,v_2',\ldots,v_{\log t},v_{\log t}'$ and $(h^\star+1)\log t$ ``dummy vertices'' to $G^\star$. The selection vertices have weight $n$ and the dummy vertices have weight $k+n\cdot(\log t-2)$. We connect all selection vertices to a clique. Furthermore, for each $i\in[\log t]$ we connect $v_i$ and $v_i'$ to $h^\star+1$ dummy vertices (different dummy vertices for different $i$). Lastly, for each $i\in[t]$ and $j\in[\log t]$ we connect each vertex corresponding to $G_i$ to $v_j$ if the $j$th bit of the binary representation of $i$ is one, and we connect each vertex corresponding to $G_i$ to $v_j'$ if the $j$th bit of the binary representation of $i$ is zero. Here, we start counting at the least significant bit. This finishes the construction. We give an illustration in \cref{fig:nopk2}. The instance can clearly be computed in \PNP-time. Note that $(h^\star+k^\star)$ is polynomially bounded in $n+\log t$.

\contourlength{2pt}

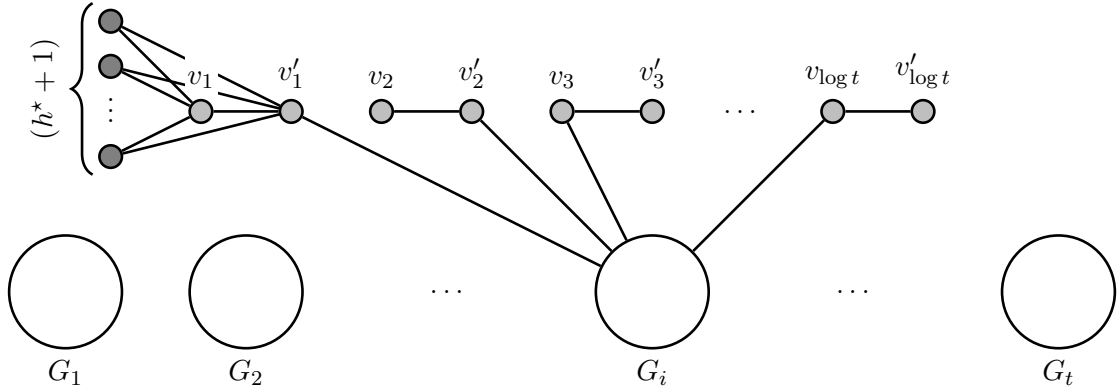
\begin{figure}[t]
\begin{center}
\begin{tikzpicture}[line width=1pt,scale=1.2,xscale=1]
\node (dv1) at (-.5,2) {};
\node (dvv1) at (.5,2) {};
\node[vert,fill=gray] (d1) at (-1.5,3) {};
\node[vert,fill=gray] (d2) at (-1.5,2.5) {};
\node (d) at (-1.5,2.1) {\vdots};
\node[vert,fill=gray] (dh) at (-1.5,1.5) {};

    \draw[decorate,decoration={brace,amplitude=8pt}] (-1.7,1.3) -- (-1.7,3.2) node [black,midway,rotate=90,yshift=18pt] {$(h^\star+1)$};

\draw (d1) --node[timelabel] {\phantom{x}} (dvv1);
\draw (d2) --node[timelabel] {\phantom{x}} (dvv1);
\draw (dh) -- (dvv1);

\node[vert,fill=lightgray,label=above:$v_1$] (v1) at (-.5,2) {};
\node[vert,fill=lightgray,label=above:$v_1'$] (vv1) at (.5,2) {};
\node[vert,fill=lightgray,label=above:$v_2$] (v2) at (1.5,2) {};
\node[vert,fill=lightgray,label=above:$v_2'$] (vv2) at (2.5,2) {};
\node[vert,fill=lightgray,label=above:$v_3$] (v3) at (3.5,2) {};
\node[vert,fill=lightgray,label=above:$v_3'$] (vv3) at (4.5,2) {};
\node (d) at (5.5,2) {\ldots};
\node[vert,fill=lightgray,label=above:$v_{\log t}$] (vt) at (6.5,2) {};
\node[vert,fill=lightgray,label=above:$v_{\log t}'$] (vvt) at (7.5,2) {};

\draw (d1) -- (v1);
\draw (d2) -- (v1);
\draw (dh) -- (v1);

\draw (v1) -- (vv1);
\draw (v2) -- (vv2);
\draw (v3) -- (vv3);
\draw (vt) -- (vvt);

\node[vert,label=below:$G_1$,minimum width=1.5cm] (g1) at (-2,0) {};
\node[vert,label=below:$G_2$,minimum width=1.5cm] (g2) at (0,0) {};
\node (d) at (2.25,0) {\ldots};
\node[vert,label=below:$G_i$,minimum width=1.5cm] (gi) at (4.5,0) {};
\node (d) at (6.75,0) {\ldots};
\node[vert,label=below:$G_t$,minimum width=1.5cm] (gt) at (9,0) {};

\draw (gi) -- (vv1);
\draw (gi) -- (vv2);
\draw (gi) -- (v3);
\draw (gi) -- (vt);

\end{tikzpicture}
\end{center}
    \caption{Illustration of the construction used in the proof of  \cref{thm:nopk2}. Dark gray vertices have weight $k+n(\log t-2)$ and light gray vertices have weight $n$. The dummy vertices are only depicted for~$v_1,v_1'$. 
    The selection vertices form a clique, none of the edges are depicted except between pairs of selection vertices of which exactly one needs to be deleted.
    In the example, we have that the first and second bit of the binary representation of $i$ is a zero, the third is a one, and the last is a one as well.}\label{fig:nopk2}
\end{figure}

    $(\Rightarrow)$: Assume that $(G^\star,w,h^\star,k^\star)$ is a \yes-instance of \textsc{Weighted Clique-Free Vertex Deletion}. Let $X$ be a deletion set of weight at most $h^\star$ such that $G^\star-X$ does not contain a clique of weight $k^\star$. First, note that for each $i\in[\log t]$ we have that one of $v_i$ and $v_i'$ need to be in $X$, since there are $h^\star+1$ cliques of weight $k^\star$ in $G^\star$ that intersect exactly in $v_i$ and $v_i'$. Since $v_i$ and $v_i'$ each have weight $n$ and $n\cdot (1+\log t)>h^\star$, we  have that exactly one of $v_i$ and $v_i'$ is contained in $X$.

    Let $i\in[t]$ be such that the $j$th bit of the binary representation of $i$ is one if and only if $v_j\notin X$. We claim that $(G_i,h,k)$ (instance $x_i$) is a \yes-instance of \textsc{Clique-Free Vertex Deletion}. First, note that $\log t$ selection vertices have weight $n\cdot\log t$, hence $X$ contains at most $h$ vertices corresponding to $G_i$. Let those vertices be $X_i$. We claim that $G_i-X_i$ does not contain any clique of size $k$. Assume for contradiction that it does. Then then every vertex in this clique is connected to $\log t$ selection vertices in $G^\star-X$. Since all selection vertices are pairwise connected, this implies that there is a clique of weight $k+n\cdot\log t$ in $G^\star-X$, a contradiction.
    
    $(\Leftarrow)$: Assume that $(G_i, h, k)$ with $i\in [t]$ is a \yes-instance of \textsc{Clique-Free Vertex Deletion}. Let $X$ be a deletion set of size at most $h$ such that $G_i-X$ does not contain any clique of size $k$.

    Now define $X'$ as follows. We add $v_j$ to $X'$ if the $j$th bit of the binary representation of $i$ is one, otherwise, we add $v_j'$ to $X'$. We claim that $X^\star=X\cup X'$ is a deletion set of weight at most $h^\star$ for~$G^\star$, that is, $G^\star-X^\star$ does not contain any clique of weight $k^\star$.

    First, note that for each $j\in[\log t]$, either $v_j$ or $v_j'$ is contained in $X^\star$. Hence, all cliques only consisting of selection vertices or dummy vertices in $G^\star-X^\star$ have weight less than $k^\star$. Furthermore, we have that for every $i'\neq i$, each vertex of $G_{i'}$ is connected to at most $\log t-1$ selection vertices in $G^\star-X^\star$, and hence, and maximum weight clique that involved vertices corresponding to $G_{i'}$ has weight at most $n\cdot \log t<h^\star$. 
    Now assume for contradiction that $G^\star-X^\star$ contains a clique of weight~$k^\star$. By the arguments above, we know that this clique consists of selection vertices and vertices corresponding to $G_i$. Furthermore, it contains at most $\log t$ selection vertices, which give weight of $n\cdot \log t$. This means that the vertices corresponding to $G_i$ in $G^\star-X^\star$ that are contained in this clique have weight $k$. This implies that there is a clique of size $k$ in $G_i-X$, a contradiction.
\end{proof}

\paragraph{Open Questions.} As already mentioned, we leave open whether \textsc{Clique-Free Vertex Deletion} admits a polynomial \OKernel when parameterized by $(h+k)$. Furthermore, our parameterized hardness results do not rule out that \textsc{Clique-Free Vertex Deletion} may be in \FPTNP when parameterized solely by $h$ or solely by $k$, which we also leave open for future research.

\section{Application to Discovery Problems}\label{sec:discovery}

We call a graph problem a ``discovery problem'' if the graph structure needs to be ``discovered'' by solving instances of an \NP-complete problem. More formally, for every pair of vertices $u,v$ in the input graph, we have an instance $I_{u,v}$ of an \NP-complete problem $L$. If $I_{u,v}$ is a \yes-instance of $L$, then there is an edge between $u$ and $v$, and otherwise there is not. Given a vertex set $V$ and for each pair $u,v\in V$ of vertices an instance $I_{u,v}$ of an \NP-complete problem, say \textsc{Satisfiability}, we call the graph $G=(V,E)$ where $\{u,v\}\in E$ if and only if $I_{u,v}$ is a \yes-instance the \emph{discovered graph}.

In this section, we first apply our framework to the discovery version of \textsc{Vertex Cover Reconfiguration}~\cite{ito2011complexity,mouawad2017parameterized,mouawad2018vertex}. Then in the second part, we discuss how to apply our framework to discovery problems in a general way.

\subsection{\textsc{Discovery Vertex Cover Reconfiguration}}

\emph{Reconfiguration problems} capture the setting where we wish to find a step-by-step transformation between two feasible (or minimal) solutions of a problem such that all intermediate results are also feasible. More specifically, given a graph $G=(V,E)$ and two solutions $S,T$ for some problem $P$, we wish to find a sequence of bounded length of solutions of bounded size, such that consecutive solutions have a small symmetric difference\footnote{We use the symbol $\triangle$ for the symmetric difference. For two sets $S,T$ we have $S\triangle T=(S\cup T)\setminus (S\cap T)$.}, say one (which models that a vertex can be either added or removed).

In this section, we investigate the following problem.

\problemdef{\textsc{Discovery Vertex Cover Reconfiguration}}{A vertex set $V$ and for each pair $u,v\in V$ of vertices an instance $I_{u,v}$ of \textsc{Satisfiability}, two minimal vertex covers $S,T\subseteq V$ for the discovered graph $G=(V,E)$, and two integers $k,\ell$.}{Does there exist a sequence of $\ell$ vertex covers $X_1,\ldots,X_\ell$ each of size at most $k$ for the discovered graph $G$ such that $X_1=S$, $X_\ell=T$, and for each $i\in[\ell-1]$ we have $|X_i\triangle X_{i+1}|\le 1$?}

\paragraph{Parameterizations.}

 This problem has several natural parameters that we will consider:
\begin{itemize}
    \item the size $\ell$ of the reconfiguration sequence,
    \item the size bound $k$ for the solutions, and 
    \item the number $n$ of vertices.
\end{itemize}
We clearly can assume that $k\le n$. Furthermore, we also can assume that $\ell\le 2^n$ since we can upper-bound the number of different vertex covers by $2^n$ and we can assume that no vertex cover is required to appear twice in the reconfiguration sequence.
\begin{theorem}\label{thm:vcdiscover}
    The problem \textsc{Discovery Vertex Cover Reconfiguration} is
    \begin{itemize}
        \item \PSPACE-hard,
        \item \paracoNP-hard when parameterized $n$, and
        \item in \FPTNP when parameterized by $k$.
    \end{itemize}
\end{theorem}

\begin{proof}
    The problem \textsc{Vertex Cover Reconfiguration}, which is the variant of \textsc{Discovery Vertex Cover Reconfiguration} where the input graph $G=(V,E)$ is given explicitly, is known to be \PSPACE-complete~\cite{ito2011complexity}. From this we immediately get the first result of \cref{thm:vcdiscover}. Furthermore, \textsc{Vertex Cover Reconfiguration} is known to be fixed-parameter tractable when parameterized by the size $k$ of a minimum vertex cover~\cite{mouawad2017parameterized}. From this, we can obtain the last result of \cref{thm:vcdiscover} by first discovering the graph using the \NP-oracle and then applying the \FPT-algorithm by \citet{mouawad2017parameterized}.

    The second result, intuitively, follows from the \NP-hardness of discovering the graph. We give a reduction from the \NP-hard problem \textsc{Satisfiability}~\cite{Karp72} and show that the produced \textsc{Discovery Vertex Cover Reconfiguration} instance is a \yes-instance if and only if the original \textsc{Satisfiability} instance is a \no-instace. Given an instance $\phi$ of \textsc{Satisfiability} we create an instance of \textsc{Discovery Vertex Cover Reconfiguration} as follows.
    We create a set of four vertices $V=\{a,b,c,d\}$. We set $I_{a,b}$, $I_{b,c}$, and $I_{c,d}$ to trivial a \yes-instance of \textsc{Satisfiability} (of constant size). We set $I_{a,d}$ to $\phi$. For all remaining vertex pairs $u,v\in V$, we set $I_{u,v}$ to a trivial \no-instance of \textsc{Satisfiability} (of constant size). We set $S=\{a,c\}$, $T=\{b,d\}$, $k=3$, and $\ell=5$. Note that $S$ and $T$ are minimal vertex covers for the discovered graph $G$, independently from whether $\phi$ is a \yes-instance of \textsc{Satisfiability} or not. For an illustration see \cref{fig:vchard} (right side).

    $(\Rightarrow)$: Assume that the produced \textsc{Discovery Vertex Cover Reconfiguration} instance is a \yes-instance. Assume for contradiction that $\phi$ is also a \yes-instance of \textsc{Satisfiability}. Then we have that the discovered graph $G$ has edges $\{a,b\}$, $\{b,c\}$, $\{c,d\}$, and $\{a,d\}$. Consider a set $X_2$ with $|S\triangle X_2|\le 1$ and $X_2\neq S$. If $|X_2|=1$, then $X_2$ is clearly not a vertex cover for $G$. Hence, we have that $|X_2|=3$. Assume that $X_2=\{a,b,c\}$ (the case where $X_2=\{a,c,d\}$ is analogous). 
    Now we have for any set $X_3$ with $|X_2\triangle X_3|\le 1$ and $|X_3|\le 3$ that either $X_3=S$ ($b$ is removed) or that edge $X_3$ is not a vertex cover (edge $\{a,d\}$ is not covered if $a$ is removed and edge $\{c,d\}$ is not covered if $c$ is removed).
    If follows that there is no reconfiguration sequence from $S$ to $T$ and we have reached a contradiction.

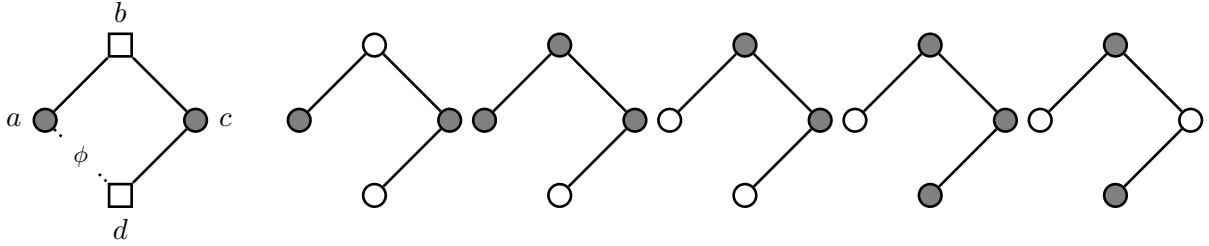
\begin{figure}[t]
\begin{center}
\begin{tikzpicture}[line width=1pt,scale=1,xscale=1]
\node[vert,label=left:$a$, fill=gray] (a) at (0,0) {};
\node[vert,label=above:$b$, rectangle] (b) at (1,1) {};
\node[vert,label=right:$c$, fill=gray] (c) at (2,0) {};
\node[vert,label=below:$d$, rectangle] (d) at (1,-1) {};

\draw (a) -- (b);
\draw (b) -- (c);
\draw (c) -- (d);
\draw[dotted] (a) --node[timelabel] {$\phi$} (d);

\end{tikzpicture}
\begin{tikzpicture}[line width=1pt,scale=1,xscale=1]
\phantom{
\node[vert,label=left:$a$] (a) at (0,0) {};
\node[vert,label=above:$b$] (bp) at (1,1) {};
\node[vert,label=below:$d$] (dp) at (1,-1) {};}
\node[vert, fill=gray] (a) at (0,0) {};
\node[vert] (b) at (1,1) {};
\node[vert, fill=gray] (c) at (2,0) {};
\node[vert] (d) at (1,-1) {};

\draw (a) -- (b);
\draw (b) -- (c);
\draw (c) -- (d);
\end{tikzpicture}
\begin{tikzpicture}[line width=1pt,scale=1,xscale=1]
\phantom{
\node[vert,label=above:$b$] (bp) at (1,1) {};
\node[vert,label=below:$d$] (dp) at (1,-1) {};}
\node[vert, fill=gray] (a) at (0,0) {};
\node[vert, fill=gray] (b) at (1,1) {};
\node[vert, fill=gray] (c) at (2,0) {};
\node[vert] (d) at (1,-1) {};

\draw (a) -- (b);
\draw (b) -- (c);
\draw (c) -- (d);
\end{tikzpicture}
\begin{tikzpicture}[line width=1pt,scale=1,xscale=1]
\phantom{
\node[vert,label=above:$b$] (bp) at (1,1) {};
\node[vert,label=below:$d$] (dp) at (1,-1) {};}
\node[vert] (a) at (0,0) {};
\node[vert, fill=gray] (b) at (1,1) {};
\node[vert, fill=gray] (c) at (2,0) {};
\node[vert] (d) at (1,-1) {};

\draw (a) -- (b);
\draw (b) -- (c);
\draw (c) -- (d);
\end{tikzpicture}
\begin{tikzpicture}[line width=1pt,scale=1,xscale=1]
\phantom{
\node[vert,label=above:$b$] (bp) at (1,1) {};
\node[vert,label=below:$d$] (dp) at (1,-1) {};}
\node[vert] (a) at (0,0) {};
\node[vert, fill=gray] (b) at (1,1) {};
\node[vert, fill=gray] (c) at (2,0) {};
\node[vert, fill=gray] (d) at (1,-1) {};

\draw (a) -- (b);
\draw (b) -- (c);
\draw (c) -- (d);
\end{tikzpicture}
\begin{tikzpicture}[line width=1pt,scale=1,xscale=1]
\phantom{
\node[vert,label=above:$b$] (bp) at (1,1) {};
\node[vert,label=below:$d$] (dp) at (1,-1) {};}
\node[vert] (a) at (0,0) {};
\node[vert, fill=gray] (b) at (1,1) {};
\node[vert] (c) at (2,0) {};
\node[vert, fill=gray] (d) at (1,-1) {};

\draw (a) -- (b);
\draw (b) -- (c);
\draw (c) -- (d);
\end{tikzpicture}
\end{center}
    \caption{Illustration of the construction used in the last part of the proof of \cref{thm:vcdiscover}. The left side shows the produced instance, where the gray vertices are the starting vertex cover, the square-shaped vertices are the target vertex cover, and edge $\{a,d\}$ is present if and only if $\phi$ is a \yes-instance. The right side shows the reconfiguration sequence constructed if $\phi$ is a \no-instance, where the gray vertices are contained in the respective vertex covers.}\label{fig:vchard}
\end{figure}

    $(\Leftarrow)$: Assume that $\phi$ is a \no-instance of \textsc{Satisfiability}. Then we have that the discovered graph $G$ has edges $\{a,b\}$, $\{b,c\}$, and $\{c,d\}$. Consider the sets $X_2=\{a,b,c\}$, $X_3=\{b,c\}$, and $X_4=\{b,c,d\}$. It is easy to verify that each set is a vertex cover of size at most 3 for $G$. Furthermore, we have $|X_i\triangle X_{i+1}|\le 1$ for each $i\in[4]$, where $X_1=S$ and $X_5=T$. Hence, we have that $S,X_2,X_3,X_4,T$ is a configuration sequence of length $5$ from $S$ to $T$. For an illustration see \cref{fig:vchard} (left side). It follows that the produced \textsc{Discovery Vertex Cover Reconfiguration} instance is a \yes-instance.
\end{proof}

\paragraph{Kernelization.} \cref{thm:vcdiscover} raises the natural question of whether \textsc{Discovery Vertex Cover Reconfiguration} admits a polynomial \OKernel when parameterized by $k$. We answer this question positively. \citet{mouawad2017parameterized} showed that \textsc{Vertex Cover Reconfiguration} admits a ``polynomial reconfiguration kernel'' when parameterized by $k$. We remark that they use a slightly different problem definition, where they do not require the start and target vertex covers to be minimal, which seems to be the main reason they do not show that the problem admits a regular polynomial kernel. For details, we refer to their paper~\cite{mouawad2017parameterized}. Informally, polynomial reconfiguration kernels are more similar to polynomial Turing kernels~\cite[Chapter~22]{Fom+19}, in that the kernelization algorithm produces a set of instances whose size is bounded polynomially in the parameter, and the original instance is required to be a \yes-instance if and only one of the instances in the set is a \yes-instance.  Here, we show how to adapt the ideas of \citet{mouawad2017parameterized} to obtain a polynomial \OKernel for \textsc{Discovery Vertex Cover Reconfiguration} parameterized by $k$.

\begin{theorem}\label{thm:vckernel}
    The problem \textsc{Discovery Vertex Cover Reconfiguration} admits a polynomial \OKernel when parameterized by $k$.
\end{theorem}
\begin{proof}
We first use the \NP-oracle to compute the discovered graph $G$ in \PNP-time. Then, we use a kernelization algorithm by \citet{damaschke2009union} that produces a graph $G'$ of size $\OO(k^2)$ that contains all minimum vertex covers of size at most $k$ of $G$.
\citet{mouawad2017parameterized} showed that if there is a reconfiguration sequence from $S$ to $T$, then there is also a reconfiguration sequence from $S$ to $T$ that only uses vertices that are contained in some minimal vertex cover, and hence are contained in~$G'$.
We produce an instance of \textsc{Discovery Vertex Cover Reconfiguration} as follows. We use the vertex set of $G'$ and for each pair $u,v$ of vertices we set $I_{u,v}$ to a trivial \yes-instance (of constant size) of \textsc{Satisfiability} if $\{u,v\}$ is an edge in $G'$; otherwise, we set $I_{u,v}$ to a trivial \no-instance (of constant size) of \textsc{Satisfiability}. We leave $k$ and $\ell$ as they are. The produced instance clearly has size $\OO(k^2)$ and is computed in \PNP-time. 
If the reduced instance is a \yes-instance, then clearly the original instance is one as well, since any solution for the reduced instance is also a solution for the original instance.
If the original instance is a \yes-instance, then there exists a reconfiguration sequence that only uses vertices from $G'$~\cite{mouawad2017parameterized} and hence the reduced instance is also a \yes-instance.
\end{proof}

\paragraph{Open Questions.} We leave the parameterized complexity of \textsc{Discovery Vertex Cover Reconfiguration} when parameterized by the configuration length $\ell$ open. We remark that \textsc{Vertex Cover Reconfiguration} is known to be \wone-hard when parameterized by $\ell$~\cite{mouawad2018vertex}.

\subsection{General Discovery Problems}

Note that the approach from \cref{thm:vckernel}, while being straightforward, is also very general. Informally, whenever we have a graph problem $L$ that admits a polynomial kernel for some parameterization, the ``discovery version'' for that problem admits a polynomial \OKernel for the same parameterization. This is interesting if $L$ is presumably not in \PNP (that is, e.g.\ \SigmaPtwo-hard or \PiPtwo-hard), since otherwise the \OKernel can solve the problem and produce trivial equivalent instances of constant size. In the following, we formalize this idea and show that it is also applicable to non-graph problems.

Formally, we call a problem $L$ a \emph{graph problem} if it has a graph as part of the input, and the input size is polynomially bounded by the size of the graph. We call a problem $L$ a \emph{Boolean problem} if it has a Boolean formula (in CNF or DNF) as part of the input, and the input size is polynomially bounded by the size of the formula. We call a problem $L$ a \emph{set system problem} if it has a family of sets (over the same universe) as part of the input, and the input size is polynomially bounded by the size of the set family.
Examples of graph problems are \textsc{Vertex Cover} and \textsc{Clique}, examples of Boolean problems are \textsc{Satisfiability} and \textsc{Tautology}, and examples of set system problems are \textsc{Hitting Set} and \textsc{Set Cover}~\cite{Karp72,GJ79}.

We have already given the main idea of the definition of a discovery graph problem at the beginning of this section. Formally, if $L$ is a graph problem, a Boolean problem, or a set system problem, then the problem \textsc{Discovery~$L$} is defined as follows.

\problembox{\textsc{Discovery~{\boldmath $L$}}}{
\begin{itemize}
    \item If $L$ is a graph problem, then instead of having a graph $G=(V,E)$ as input, \textsc{Discovery~$L$} has a set of vertices $V$ as input and for each pair of vertices $u,v\in V$, we have an instance $I_{u,v}$ of \textsc{Satisfiability}. If $I_{u,v}$ is a \yes-instance, then this is interpreted as an edge being present between $u$ and $v$ in $G$. If $I_{u,v}$ is a \no-instance, then this is interpreted as no edge being present between $u$ and $v$ in $G$.
    \item If $L$ is a Boolean problem, then instead of having a Boolean formula $\phi$ as input, \textsc{Discovery~$L$} has a set of variables $X$ and a number of clauses $C$ as input, and for each pair of a literal (negated or non-negated variable) and a clause $\ell,c\in \{x,\neg x\mid x\in X\}\times [C]$, we have an instance $I_{\ell,c}$ of \textsc{Satisfiability}. If $I_{\ell,c}$ is a \yes-instance, then this is interpreted as the literal $\ell$ appearing in clause $c$ of $\phi$. If $I_{\ell,c}$ is a \no-instance, then this is interpreted as the literal $\ell$ not appearing in clause $c$ of $\phi$.
    \item If $L$ is a set system problem, then instead of having an explicitly given family of sets $S_1, \ldots, S_n\subseteq U$ over some universe $U$ as input, \textsc{Discovery~$L$} has the universe set $U$ and a number of sets $S$ as input, and for each pair of a universe element and a set $u,s\in U\times [S]$, we have an instance $I_{u,s}$ of \textsc{Satisfiability}. If $I_{u,s}$ is a \yes-instance, then this is interpreted as the element $u$ being contained in set $s$. If $I_{u,s}$ is a \no-instance, then this is interpreted as the element $u$ not being contained in set $s$.
\end{itemize}
}

Formally, we obtain the following meta-theorem.
\begin{theorem}
    Let $L$ be a graph problem, a Boolean problem, or a set system problem. If $L$ admits a polynomial kernel when parameterized by some parameter $k$, then \textsc{Discovery~$L$} parameterized by $k$ admits a polynomial \OKernel.
\end{theorem}
\begin{proof}
    Given an instance $I$ of \textsc{Discovery~$L$} parameterized by $k$, we use an \NP-oracle to decide all \textsc{Satisfiability} instances in the input, hereby computing a graph $G$, a Boolean formula $\phi$, or a set system $S_1, \ldots, S_n\subseteq U$, respectively. With this, we can create an equivalent instance $I'$ of~$L$ parameterized by $k$. Now we apply the kernelization algorithm for $L$ to $I'$, thereby creating an equivalent instance $I''$ with $|I''|\in k^{\OO(1)}$. We transform the instance $I''$ of $L$ to an equivalent instance~$I'''$ of \textsc{Discovery~$L$} by using trivial \yes- and \no-instances (of constant size) of \textsc{Satisfiability} in a straightforward manner.
    This finishes the description of the polynomial \OKernel algorithm. It clearly runs in \PNP-time and we clearly have that $|I'''|\in k^{\OO(1)}$. The correctness follows from the correctness of the kernelization algorithm for $L$ parameterized by $k$.
\end{proof}

\section{Conclusion and Future Work}
Motivated by the success of SAT-solvers and ILP-solvers, we introduced a new concept for kernels, called \OKernel{s}, that allows the kernelization algorithms to use \NP-oracles. Furthermore, we presented a method to refute such \OKernel{s} of polynomial size. To illustrate its use case, we applied our framework to three problem settings. For each one, we gave specific open questions in the respective section.

There are several natural future work directions to expand this framework. Since \NP-oracle calls are still much more computationally expensive than basic operations, one may want to limit the number of oracle calls that the kernelization algorithms may use. Natural restrictions would be to allow only a logarithmic number or a constant number of oracle calls. On the flip side, one might consider giving the kernelization algorithm access to more powerful oracles, such as e.g.~\SigmaPtwo-oracles. A natural extension of the composition framework would be to provide \PNP-AND-cross-compositions. They can be defined in an analogous way to the \PNP-OR-cross-compositions, and we conjecture that the result for the classical AND-cross-compositions can be adapted to the oracle setting.

Finally, there are many further problems to which our framework can be applied. Indeed, we can, essentially, pick any parameterized problem that is unlikely to be in \PNP or in \FPT, and check whether it admits a \OKernel.

\bibliographystyle{abbrvnat}
\bibliography{bib}	

\end{document}